\documentclass[conference]{IEEEtran}
\IEEEoverridecommandlockouts

\usepackage{cite}
\usepackage{algorithmic}
\usepackage{graphicx}
\usepackage{textcomp}
\usepackage{xcolor}
\usepackage{balance}
\usepackage[cmex10]{amsmath}
\usepackage{amsmath}
\usepackage{amssymb}
\usepackage{amsthm}
\usepackage{amsfonts}
\usepackage{bm}
\usepackage{xfrac}
\usepackage{empheq}
\usepackage[normalem]{ulem} 
\usepackage{soul} 
\usepackage{mathtools}
\newtheorem{theorem}{Theorem}

\newtheorem{proposition}{Proposition}

\newtheorem{corollary}{Corollary}
\newtheorem{remark}{Remark}
\theoremstyle{definition}
\newtheorem{definition}{Definition}
\def\BibTeX{{\rm B\kern-.05em{\sc i\kern-.025em b}\kern-.08em
    T\kern-.1667em\lower.7ex\hbox{E}\kern-.125emX}}

\usepackage{enumitem}
\usepackage{microtype}
\usepackage{booktabs}
\usepackage{graphicx}
\usepackage{array}
\usepackage{bbm}

\newcommand{\X}{\mathcal X}
\newcommand{\Y}{\mathcal Y}
\newcommand{\1}{\mathbbm 1}

\newcommand{\Z}{\mathbb Z}
\newcommand{\E}{\mathbb E}

\newcommand{\diam}{\operatorname{diam}}
\newcommand{\pos}[1]{\left[#1\right]_+}
\begin{document}

\title{
Sparse Discrete Laplace and Gaussian Mechanisms under Local Differential Privacy 
}

\author{\IEEEauthorblockN{1\textsuperscript{st} Amirreza Zamani}
\IEEEauthorblockA{\textit{Department of Information Science and Engineering, KTH} \\
Stockholm, Sweden \\
amizam@kth.se}
\and
\IEEEauthorblockN{2\textsuperscript{nd} Sajad Daei}
\IEEEauthorblockA{\textit{Department of Information Science and Engineering, KTH} \\
Stockholm, Sweden \\
sajado@kth.se}
\and
\IEEEauthorblockN{3\textsuperscript{rd} Parastoo Sadeghi}
\IEEEauthorblockA{\textit{School of Engineering and Technology, UNSW} \\
Canberra, Australia \\
p.sadeghi@unsw.edu.au}
\and
\IEEEauthorblockN{4\textsuperscript{th} Mikael Skoglund}
\IEEEauthorblockA{\textit{Department of Information Science and Engineering, KTH} \\
Stockholm, Sweden \\
skoglund@kth.se}}
\maketitle
\begin{abstract}
We study sparse locally private channels of the form
$
M(y\mid x)\propto w(x,y) \1\{y\in S(x)\},
$
where the admissible output set $S(x)$ is allowed to depend on the private input $x$ and is assumed to be small. Here, we consider the sparse discrete-Laplace family with kernel $w(x,y)=e^{-\lambda d(x,y)}$ and the sparse Gaussian family with kernel $w(x,y)=e^{-d(x,y)^2/(2\sigma^2)}$. For both families we give exact characterizations of pure and approximate local differential privacy. For pure $\varepsilon$-local differential privacy, we show that input-dependent sparse supports are obtained when all supports coincide. For $(\varepsilon,\delta)$-local differential privacy, we derive exact formulas for the privacy defect in terms of support leakage and excess privacy loss on the overlap region. We then specialize the analysis to radius-truncated sparse discrete-Laplace and radius-truncated sparse Gaussian mechanisms and obtain explicit privacy-sparsity tradeoffs in terms of the support size $s$. In particular, we show that nontrivial approximate local privacy requires a minimum support size, whereas larger supports reduce support leakage but increase distortion. For the Gaussian family, the overlap term exhibits an additional quadratic dependence on the support radius, which implies a sharper tradeoff between privacy and sparsity. These results identify the support cardinality as the intrinsic complexity parameter of the mechanism and yield an optimal design principle: choose the smallest support size that satisfies the target privacy constraint.
\end{abstract}

\section{Introduction}

Differential privacy was introduced as a formal framework to limit the information leaked by the output of a randomized algorithm about any single input record \cite{dwork2006calibrating,dwork2014book}. In the \emph{local} model, each user privatizes its data before any collection takes place, thereby removing the need for a trusted curator. This model has become central in both theory and deployment, with classical examples including randomized-response-style mechanisms and practical systems  \cite{dwork2014book,rappor2014}.

However, considering design problems with discrete alphabets, many local randomizers spread probability mass over the entire output alphabet. This can be statistically inefficient when only a small set of outputs are plausible surrogates for a given input. Motivated by this observation, we study sparse families of locally private channels in which each input $x$ is allowed to be randomized only over a small admissible output set $S(x)$.

The support restriction serves three roles. First, it reflects the geometry of the discrete problem by concentrating probability on outputs that remain close to the private input. Second, it reduces distortion by excluding far-away outputs that contribute little to utility but substantially to error. Third, under approximate local privacy, the support mismatch between different inputs gives a natural interpretation of the $\delta$-budget: outputs that can occur under one input but not another are precisely the exceptional events accounted for by $\delta$.

In this paper, we focus on two sparse families. The first is the sparse discrete-Laplace family
\[
Q_\lambda(y\mid x)=\frac{e^{-\lambda d(x,y)}}{Z_x}\1\{y\in S(x)\},
,
\]
where $Z_x=\sum_{y'\in S(x)} e^{-\lambda d(x,y')}$ and $\lambda>0$ is an inverse-temperature parameter. The second is the sparse Gaussian family
\[
G_\sigma(y\mid x)=\frac{e^{-d(x,y)^2/(2\sigma^2)}}{W_x}\1\{y\in S(x)\},
,
\]
where $W_x=\sum_{y'\in S(x)} e^{-d(x,y')^2/(2\sigma^2)}$.
In both cases, $d$ is a discrete metric or distortion function and $S(x)\subseteq \Y$ is a support set whose cardinality plays the role of the sparsity level.

\paragraph{Positioning and novelty}
Sparsity under local privacy has already been studied at the \emph{estimation level}. In particular, sparse discrete distribution estimation under $\varepsilon$-LDP has minimax sample complexity governed by the support size rather than only by the ambient alphabet size \cite{acharya2021sparse}. Likewise, Laplace-, geometric-, and Gaussian-type perturbation mechanisms are classical objects in the privacy literature, including metric-based variants such as geo-indistinguishability and local counting mechanisms based on geometric noise \cite{andres2013geo,kacem2018geometric,balle2018gaussian,canonne2020discretegaussian}. The likely new part of the present note is therefore not sparsity under LDP in general, nor the use of Laplace/Gaussian kernels themselves, but rather the explicit study of \emph{mechanism-level sparsity} through input-dependent support sets $S(x)$ and the resulting privacy-sparsity tradeoffs. To the best of our knowledge, this exact theorem-level treatment of sparse-support local channels is not standard in the existing local privacy literature.

\paragraph{Contributions}
Our contributions are as follows.
\begin{enumerate}[label=(\roman*)]
    \item For discrete-Laplace and Gaussian families with sparse support, we obtain exact characterizations of pure $\varepsilon$-LDP and approximate $(\varepsilon,\delta)$-LDP.
    \item We show that input-dependent sparse supports with pure local privacy are obtained when all supports coincide.
    \item For radius-truncated sparse mechanisms, we derive explicit formulas for the privacy defect and decompose it into a support-leakage term and an overlap excess-loss term.
    \item We show that the support size $s=|S(x)|$ acts as an intrinsic complexity parameter of the channel: nontrivial approximate privacy requires a minimum level of support overlap, while larger supports reduce leakage but increase distortion.
    \item We derive mechanism-design rules showing that the distortion-optimal sparse-support mechanism is obtained by choosing the smallest support size compatible with the target privacy constraint.
\end{enumerate}

\subsection{Related work}

At the level of the privacy framework, differential privacy was introduced by
Dwork, McSherry, Nissim, and Smith~\cite{dwork2006calibrating}, and the local
model is discussed extensively in the monograph of Dwork and
Roth~\cite{dwork2014book}. Practical deployment examples include
RAPPOR~\cite{rappor2014}, while, from a historical perspective, the
local model can be viewed as a formal descendant of randomized response,
originally introduced by Warner~\cite{warner1965randomized}.

At the level of the statistical problem, sparse estimation under local privacy
is already well understood in several important settings. The most directly
relevant work for our purposes is the sample-complexity characterization of
sparse discrete distribution estimation under LDP by Acharya, Kairouz, Liu, and
Sun~\cite{acharya2021sparse}, which shows that the intrinsic support size,
rather than only the ambient alphabet size, governs the difficulty of the
estimation problem. More broadly, minimax lower and upper bounds for locally
private estimation have been established by Duchi, Jordan, and
Wainwright~\cite{duchi2014local,duchi2018minimax}, highlighting the statistical
cost of local privacy and providing a general information-theoretic backdrop for
mechanism design under LDP.

At the level of noise families and metric-based mechanisms, our work builds on
several lines. Metric-based Laplace mechanisms and their privacy--utility
trade-offs have been developed in the geo-indistinguishability
literature~\cite{andres2013geo}. In the local model, Kacem and Palamidessi
study geometric noise for locally private counting queries, including truncated
variants~\cite{kacem2018geometric}. On the Gaussian side, Balle and Wang
analyze the Gaussian mechanism using exact analytical
calibration~\cite{balle2018gaussian}, while Canonne, Kamath, and Steinke show
that the discrete Gaussian matches the privacy and utility behavior of the
continuous Gaussian while being exactly samplable on finite
computers~\cite{canonne2020discretegaussian}.

Several additional works motivate broader discrete mechanism classes beyond the
Laplace and Gaussian families considered explicitly in this paper. McSherry and
Talwar~\cite{mcsherry2007mechanism} introduced the exponential mechanism, which
provides a general score-based way to sample from a discrete range under
differential privacy and serves as a natural conceptual template for
sparse-support kernels of the form
$
M(y\mid x)\propto \exp(\eta u(x,y))\,\mathbf 1\{y\in S(x)\}.
$
Geng and Viswanath~\cite{geng2012optimal} showed that staircase-shaped noise is
optimal for a broad class of utility or cost objectives in the central model,
suggesting sparse staircase families as a natural alternative to sparse
discrete-Laplace channels. In the local model, Kairouz, Oh, and
Viswanath~\cite{kairouz2014extremal} proved that staircase mechanisms are
extremal for a broad class of information-theoretic utility functions under
$\varepsilon$-LDP, thereby providing a general optimization perspective closely
related to our support-constrained viewpoint. Their subsequent work on discrete
distribution estimation under local privacy~\cite{kairouz2016discrete}
established the order-optimality of $k$-ary randomized response and RAPPOR in
different privacy regimes, making sparse or neighborhood-restricted
randomized-response mechanisms a relevant point of comparison for our
framework.

More recent discrete-noise work further broadens the design space. Agarwal,
Kairouz, and Liu~\cite{agarwal2021skellam} analyzed the Skellam mechanism and
showed that it can achieve Gaussian-like privacy--accuracy trade-offs in
distributed settings, while Chen, \"Ozg\"ur, and
Kairouz~\cite{chen2022poisson} introduced the Poisson binomial mechanism, a
bounded-support discrete mechanism whose communication and support behavior are
especially relevant to mechanism-level sparsity. Finally, our work is also
related in spirit to utility-optimized local privacy mechanisms. In particular,
Murakami et al.~\cite{murakami2019utility} study local mechanisms explicitly
optimized for estimation utility under privacy constraints, which is
complementary to the present perspective.

Practical statistical disclosure systems also provide motivation for studying sparse and
finite-support mechanisms. In particular, Sadeghi and
Chien~\cite{Sadeghi_Chien_2024} analyze the perturbation methodology used by
the Australian Bureau of Statistics (ABS), whose disclosure-control framework
employs bounded and sparse perturbation mechanisms for counting queries in the
centralized differential privacy setting. They show, via a Lagrangian
characterization under symmetric support constraints, that the resulting optimal
perturbation distribution takes a discrete-Gaussian form, and they further
analyze the associated privacy-variance trade-offs through the corresponding
$(\epsilon,\delta)$ guarantees. This work focuses on centralized
privacy for counting queries and it provides practical evidence that sparse,
finite-support, and truncated discrete-Gaussian mechanisms arise naturally in
real-world statistical disclosure systems, thereby further motivating the
mechanism-level sparsity perspective studied in this paper.

The key distinction of our paper is that sparsity is imposed directly at the
level of the local channel itself, through input-dependent support sets
$S(x)$. Thus, rather than focusing only on estimation rates or on globally
supported perturbation families, we study the privacy, utility, and structural
consequences of \emph{mechanism-level sparsity}.


\section{Preliminaries}

Let $\X\subseteq \Z$ be a finite input alphabet and let $\Y$ be a discrete output alphabet. A channel $Q(\cdot\mid x)$ is said to be pure $\varepsilon$-locally differentially private if for all $x,x'\in \X$ and all measurable sets $A\subseteq \Y$,
\[
Q(A\mid x)\le e^{\varepsilon}Q(A\mid x').
\]
It is said to be $(\varepsilon,\delta)$-locally differentially private if
\[
Q(A\mid x)\le e^{\varepsilon}Q(A\mid x')+\delta
\]
for all $x,x'\in \X$ and all measurable sets $A\subseteq \Y$.

Since $\Y$ is discrete throughout the paper, $(\varepsilon,\delta)$-LDP is equivalent to the hockey-stick divergence bound
\[
\sum_{y\in \Y} \pos{Q(y\mid x)-e^{\varepsilon}Q(y\mid x')}\le \delta
\qquad \text{for all }x,x'\in \X,
\]
where $[x]_+=\max\{0,x\}$.
We use this characterization repeatedly.

\section{Sparse discrete-Laplace mechanisms}

\begin{definition}[Sparse discrete-Laplace family]\label{def:sparse-family}
Let $d:\X\times\Y\to [0,\infty)$ be a distortion function, typically $d(x,y)=|x-y|$. For each $x\in\X$, let $S(x)\subseteq \Y$ be a nonempty support set. For $\lambda>0$, define
\[
Q_\lambda(y\mid x)=\frac{e^{-\lambda d(x,y)}}{Z_x}\1\{y\in S(x)\},
\]
where $Z_x=\sum_{y'\in S(x)} e^{-\lambda d(x,y')}$.
We call this a \emph{sparse discrete-Laplace mechanism}. Its sparsity is governed by the cardinality $|S(x)|$. In fact, we assume that $|S(x)|$ is small relative to $|\mathcal{Y}|$, i.e., $|S(x)|\ll|\mathcal{Y}|$.
\end{definition}

\subsection{Exact privacy characterization}

\begin{theorem}
\label{thm:pure}
The sparse discrete-Laplace mechanism in Definition~\ref{def:sparse-family} is pure $\varepsilon$-LDP if and only if the following conditions hold:
\begin{enumerate}[label=(\roman*)]
    \item \textbf{Common-support condition:} for every $x,x'\in \X$,
    \[
    S(x)=S(x').
    \]
    \item \textbf{Bounded overlap privacy loss:} if all supports are equal to a common set $S\subseteq \Y$, then
    \[
    \sup_{x,x'\in \X}\sup_{y\in S}
    \left[
    \lambda\bigl(d(x',y)-d(x,y)\bigr)+\log\frac{Z_{x'}}{Z_x}
    \right]
    \le \varepsilon.
    \]
\end{enumerate}
Equivalently, under common support,
\[
\sup_{x,x'\in \X}\sup_{y\in S}
\left|
\lambda\bigl(d(x',y)-d(x,y)\bigr)+\log\frac{Z_{x'}}{Z_x}
\right|
\le \varepsilon.
\]
\end{theorem}

\begin{proof}
Suppose there exist $x,x'\in \X$ and $y\in S(x)\setminus S(x')$. Then
\[
Q_\lambda(y\mid x)>0,
\qquad
Q_\lambda(y\mid x')=0,
\]
so the ratio $Q_\lambda(y\mid x)/Q_\lambda(y\mid x')$ is infinite. Hence pure $\varepsilon$-LDP fails. Therefore common support is necessary.

Now assume $S(x)=S$ for all $x\in \X$. For any $y\in S$,
\[
\frac{Q_\lambda(y\mid x)}{Q_\lambda(y\mid x')}
=
\exp\!\bigl(-\lambda d(x,y)+\lambda d(x',y)\bigr)\cdot \frac{Z_{x'}}{Z_x}.
\]
Taking logarithms gives the pointwise privacy loss
\[
L(x,x';y)=\lambda\bigl(d(x',y)-d(x,y)\bigr)+\log\frac{Z_{x'}}{Z_x}.
\]
Pure $\varepsilon$-LDP is equivalent to $L(x,x';y)\le \varepsilon$ for all $x,x',y$, which proves the first formula. Interchanging $x$ and $x'$ yields the symmetric version.
\end{proof}

\begin{remark}
\label{rem:pure-impossible}
Theorem~\ref{thm:pure} shows that sparsity cannot be encoded through genuinely input-dependent supports in the pure local privacy regime. Under pure $\varepsilon$-LDP, all support sets must coincide and hence become input-dependent; otherwise the privacy loss is infinite on support-mismatch points.
\end{remark}
\begin{corollary}
    Let $d(x,y)=|x-y|$ and $D=\max_{x,x'}|x-x'|$. A sufficient condition for bounded overlap privacy is to have
    \begin{align*}
        \lambda D+ \log\frac{Z_{x'}}{Z_x} \le \varepsilon.
    \end{align*}
    Furthermore, if $Z_{x'}=Z_{x}$ for all $x$ and $x'$, this reduces to: 
    \begin{align*}
        \lambda D \le \varepsilon.
    \end{align*}
\end{corollary}
\begin{proof}
    The proof follows by Theorem 1 and reverse triangle inequality. For all $x$, $x'$ and $y$, we have
    \begin{align*}
        |x'-y|-|x-y|\le |x-x'|.
    \end{align*}
\end{proof}
\begin{definition}[Ordered privacy defect]\label{def:ordered-defect}
For the sparse discrete-Laplace family, define the ordered privacy defect between $x,x'\in\X$ by
\[
\Delta_{\varepsilon}(x,x'):=\sum_{y\in \Y}\pos{Q_\lambda(y\mid x)-e^{\varepsilon}Q_\lambda(y\mid x')}.
\]
\end{definition}

\begin{theorem}
\label{thm:approx-general}
The sparse discrete-Laplace mechanism in Definition~\ref{def:sparse-family} is $(\varepsilon,\delta)$-LDP if and only if
\[
\sup_{x,x'\in \X}\Delta_{\varepsilon}(x,x')\le \delta.
\]
Moreover,
\begin{align*}
\Delta_{\varepsilon}(x,x')
&=
\sum_{y\in S(x)\setminus S(x')}
\frac{e^{-\lambda d(x,y)}}{Z_x}
\\
&\qquad+
\sum_{y\in S(x)\cap S(x')}
\pos{\frac{e^{-\lambda d(x,y)}}{Z_x}-e^{\varepsilon}\frac{e^{-\lambda d(x',y)}}{Z_{x'}}}.
\end{align*}
\end{theorem}

\begin{proof}
For a discrete output alphabet, $(\varepsilon,\delta)$-LDP is equivalent to
\[
\sup_{A\subseteq \Y}\bigl(Q_\lambda(A\mid x)-e^{\varepsilon}Q_\lambda(A\mid x')\bigr)\le \delta
\qquad \text{for all }x,x'\in\X.
\]
For fixed $x,x'$, the maximizing set is
\[
A^*=\{y\in\Y:Q_\lambda(y\mid x)-e^{\varepsilon}Q_\lambda(y\mid x')>0\},
\]
so the supremum equals
\[
\sum_{y\in\Y}\pos{Q_\lambda(y\mid x)-e^{\varepsilon}Q_\lambda(y\mid x')}.
\]
This proves the first statement.

To obtain the explicit formula, split the sum into three regions: $S(x)\setminus S(x')$, $S(x)\cap S(x')$, and $\Y\setminus S(x)$. On $\Y\setminus S(x)$, the contribution is zero. On $S(x)\setminus S(x')$, we have $Q_\lambda(y\mid x')=0$, so the contribution is exactly $e^{-\lambda d(x,y)}/Z_x$. On $S(x)\cap S(x')$, both probabilities are positive and the contribution is the displayed positive part.
\end{proof}

\begin{remark}
\label{rem:decomposition}
Theorem~\ref{thm:approx-general} decomposes the privacy defect into two terms,
\[
\Delta_{\varepsilon}(x,x')=\delta_{\mathrm{supp}}(x,x')+\delta_{\mathrm{ov}}(x,x'),
\]
where
\[
\delta_{\mathrm{supp}}(x,x'):=Q_\lambda\bigl(S(x)\setminus S(x')\mid x\bigr)
\]
is the support-leakage term and
\[
\delta_{\mathrm{ov}}(x,x'):=
\sum_{y\in S(x)\cap S(x')}
\pos{\frac{e^{-\lambda d(x,y)}}{Z_x}-e^{\varepsilon}\frac{e^{-\lambda d(x',y)}}{Z_{x'}}}
\]
is the overlap excess-loss term. Thus the $\delta$-budget has a precise interpretation: it accounts both for outputs available only under one input and for overlap points where the privacy loss exceeds $\varepsilon$.
\end{remark}

\begin{definition}[Pointwise privacy loss on the overlap]\label{def:pointwise-loss}
For $y\in S(x)\cap S(x')$, define
\[
L(x,x';y):=\log\frac{Q_\lambda(y\mid x)}{Q_\lambda(y\mid x')}=\lambda\bigl(d(x',y)-d(x,y)\bigr)+\log\frac{Z_{x'}}{Z_x}.
\]
Then the overlap contribution can be rewritten as
\[
\delta_{\mathrm{ov}}(x,x')
=
\sum_{y\in S(x)\cap S(x'):L(x,x';y)>\varepsilon}
Q_\lambda(y\mid x)\bigl(1-e^{\varepsilon-L(x,x';y)}\bigr).
\]
\end{definition}

\subsection{Radius truncation and sparsity tradeoffs}

\begin{definition}[Radius-truncated mechanism]\label{def:radius-r}
Let $r\in\mathbb N$ and define
\[
S_r(x):=\{y\in \Y:|x-y|\le r\}.
\]
The radius-truncated sparse discrete-Laplace mechanism is
\begin{align*}
Q_{\lambda,r}(y\mid x)&=\frac{e^{-\lambda |x-y|}}{C_r}\1\{|x-y|\le r\},
\\
C_r&:=\sum_{k=-r}^{r}e^{-\lambda|k|}=1+2\sum_{k=1}^{r} e^{-\lambda k}.
\end{align*}
The support size is
\[
s=|S_r(x)|=2r+1.
\]
\end{definition}

\begin{definition}[Separation-dependent privacy defect]\label{def:delta-h}
For $h\ge 0$, define
\[
\delta_h(\varepsilon,\lambda,r):=\sum_{y\in\Y}\pos{Q_{\lambda,r}(y\mid 0)-e^{\varepsilon}Q_{\lambda,r}(y\mid h)}.
\]
If $\X$ is finite, let
\[
D:=\max_{x,x'\in\X}|x-x'|
\]
denote its diameter.
\end{definition}

\begin{theorem}
\label{thm:radius-exact}
The radius-truncated sparse discrete-Laplace mechanism is $(\varepsilon,\delta)$-LDP if and only if
\[
\max_{0\le h\le D}\delta_h(\varepsilon,\lambda,r)\le \delta.
\]
Moreover, if $0\le h\le 2r$, then
\begin{align*}
\delta_h(\varepsilon,\lambda,r)
&=
\frac{1}{C_r}\sum_{k=-r}^{h-r-1} e^{-\lambda|k|}
\\
&\qquad+
\frac{1}{C_r}\sum_{k=h-r}^{r}
\pos{e^{-\lambda|k|}-e^{\varepsilon}e^{-\lambda|k-h|}},
\end{align*}
and if $h>2r$, then
\[
\delta_h(\varepsilon,\lambda,r)=1.
\]
\end{theorem}

\begin{proof}
Translation-invariance gives
\[
Q_{\lambda,r}(y\mid x)=Q_{\lambda,r}(y-x\mid 0),
\]
so the privacy defect depends only on the separation $h=|x-x'|$. The general characterization in Theorem~\ref{thm:approx-general} therefore reduces to
\[
\max_{0\le h\le D}\delta_h(\varepsilon,\lambda,r)\le \delta.
\]

Now compare inputs $0$ and $h$. Their supports are
\begin{align*}
S_r(0)&=\{-r,-r+1,\dots,r\},\\
S_r(h)&=\{h-r,h-r+1,\dots,h+r\}.
\end{align*}
If $h>2r$, then these sets are disjoint, so choosing $A=S_r(0)$ yields
\[
Q_{\lambda,r}(A\mid 0)=1,
\qquad
Q_{\lambda,r}(A\mid h)=0,
\]
and therefore $\delta_h(\varepsilon,\lambda,r)=1$.

Assume now $0\le h\le 2r$. Then
\begin{align*}
S_r(0)\setminus S_r(h)&=\{-r,\dots,h-r-1\},
\\
S_r(0)\cap S_r(h)&=\{h-r,\dots,r\}.
\end{align*}
Applying Theorem~\ref{thm:approx-general} gives the stated formula.
\end{proof}

\begin{corollary}
\label{cor:pure-impossible}
If $|\X|\ge 2$, then the radius-truncated mechanism is not pure $\varepsilon$-LDP for any finite $\varepsilon$.
\end{corollary}

\begin{proof}
Distinct inputs have distinct supports, so pure privacy fails by Theorem~\ref{thm:pure}.
\end{proof}

\begin{definition}[$s$-sparse discrete-Laplace mechanism]\label{def:s-sparse}
For odd integers $s\ge 1$, let
\[
t:=\frac{s-1}{2}.
\]
Define
\[
Q_{\lambda,s}(y\mid x)=\frac{e^{-\lambda|x-y|}}{C_t}\1\{|x-y|\le t\},
\
C_t=1+2\sum_{j=1}^{t} e^{-\lambda j}.
\]
Then $Q_{\lambda,s}$ is the radius-truncated mechanism with support cardinality $s$.
\end{definition}

\begin{definition}[Privacy range]\label{def:privacy-range}
Let $H\ge 0$. We say that $Q$ is $(\varepsilon,\delta)$-LDP on range $H$ if for all $x,x'\in\X$ such that $|x-x'|\le H$ and all measurable $A\subseteq\Y$,
\[
Q(A\mid x)\le e^{\varepsilon}Q(A\mid x')+\delta.
\]
When $H=D:=\diam(\X)$, this reduces to the usual all-pairs notion.
\end{definition}

\begin{theorem}
\label{thm:feasibility}
Fix a privacy range $H\ge 0$. If there exist $x,x'\in\X$ with $|x-x'|\le H$ and
\[
|x-x'|>s-1,
\]
then
\[
\sup_{A\subseteq\Y}\bigl(Q_{\lambda,s}(A\mid x)-e^{\varepsilon}Q_{\lambda,s}(A\mid x')\bigr)=1.
\]
Consequently, nontrivial $(\varepsilon,\delta)$-LDP on range $H$ with $\delta<1$ cannot be obtained unless
\[
s\ge H+1.
\]
\end{theorem}

\begin{proof}
The support of $Q_{\lambda,s}(\cdot\mid x)$ is
\[
S(x)=\{y:|x-y|\le t\},
\qquad t=\frac{s-1}{2}.
\]
If $|x-x'|>2t=s-1$, then $S(x)\cap S(x')=\varnothing$. Taking $A=S(x)$ gives
\[
Q_{\lambda,s}(A\mid x)=1,
\qquad
Q_{\lambda,s}(A\mid x')=0,
\]
so the privacy defect equals $1$.
\end{proof}

\begin{remark}
The support set $S(x)$ should be interpreted as the collection of outputs that are plausible surrogates for the private symbol $x$. Its sparsity is natural for three reasons. First, it respects the local geometry of the problem by keeping probability mass near $x$. Second, it reduces distortion by removing far-away outputs that contribute little to utility. Third, under $(\varepsilon,\delta)$-LDP, support mismatch is not forbidden but instead paid for by the $\delta$-budget, making sparse supports a natural design degree of freedom.
\end{remark}

\begin{theorem}
\label{thm:clean-bound}
Fix a privacy range $H\ge 0$. Suppose that
\[
\lambda H\le \varepsilon
\qquad\text{and}\qquad
s\ge 2H+1.
\]
Then for every $0\le h\le H$,
\[
\delta_h(\varepsilon,\lambda,s)=\frac{1}{C_t}\sum_{j=t-h+1}^{t} e^{-\lambda j},
\qquad t=\frac{s-1}{2},
\]
and hence
\begin{align*}
\delta^*(\varepsilon,\lambda,s;H):=\max_{0\le h\le H}\delta_h(\varepsilon,\lambda,s)
&\le
\frac{H e^{-\lambda(t-H+1)}}{C_t}\\
&\le H e^{-\lambda(t-H+1)}.
\end{align*}
\end{theorem}

\begin{proof}
For separation $h\le H$, the pointwise privacy loss on the overlap is
\[
L_h(k)=\log\frac{Q_{\lambda,s}(k\mid 0)}{Q_{\lambda,s}(k\mid h)}
=\lambda\bigl(|k-h|-|k|\bigr),
\]
since the normalizing constants are equal. By the reverse triangle inequality,
\[
\bigl||k-h|-|k|\bigr|\le h\le H,
\]
hence
\[
L_h(k)\le \lambda H\le \varepsilon.
\]
So the overlap positive-part term vanishes. Only the support-leakage term remains, which yields the claimed sum and bound.
\end{proof}

\begin{corollary}
\label{cor:sufficient-s}
Under the conditions of Theorem~\ref{thm:clean-bound}, a sufficient condition for $(\varepsilon,\delta)$-LDP on range $H$ is
\[
s\ge 2H-1+\frac{2}{\lambda}\log\frac{H}{\delta}.
\]
In particular, if one chooses the largest overlap-compatible value
\[
\lambda=\frac{\varepsilon}{H},
\]
then it suffices that
\[
s\ge 2H-1+\frac{2H}{\varepsilon}\log\frac{H}{\delta}.
\]
\end{corollary}

\begin{proof}
By Theorem~\ref{thm:clean-bound}, it suffices that
\[
H e^{-\lambda(t-H+1)}\le \delta.
\]
Taking logarithms yields
\[
t\ge H-1+\frac{1}{\lambda}\log\frac{H}{\delta}.
\]
Substituting $s=2t+1$ gives the claim.
\end{proof}

\begin{definition}[Distortion moments]
For the $s$-sparse discrete-Laplace mechanism $Q_{\lambda,s}$, define the first and second distortion moments by
\begin{align*}
R_1(\lambda,s)&:=\E[|Y-x|\mid X=x],
\\
R_2(\lambda,s)&:=\E[(Y-x)^2\mid X=x].
\end{align*}
By translation-invariance, these quantities do not depend on $x$.
\end{definition}

\begin{proposition}
\label{prop:distortion}
Let $t=(s-1)/2$. Then
\begin{align*}
R_1(\lambda,s)&=\frac{2\sum_{j=1}^{t} j e^{-\lambda j}}{C_t},
\\
R_2(\lambda,s)&=\frac{2\sum_{j=1}^{t} j^2 e^{-\lambda j}}{C_t}.
\end{align*}
Moreover, for fixed $\lambda$, both $R_1(\lambda,s)$ and $R_2(\lambda,s)$ are nondecreasing in $s$.
\end{proposition}

\begin{proof}
The formulas follow by symmetry. Increasing $t$ adds probability mass at larger distance levels, so both moments are nondecreasing.
\end{proof}

\begin{theorem}
\label{thm:optimal-sparsity}
Fix $\lambda>0$, a privacy range $H$, and target privacy parameters $(\varepsilon,\delta)$. Among all support sizes $s$ such that the $s$-sparse discrete-Laplace mechanism satisfies $(\varepsilon,\delta)$-LDP on range $H$, the distortion-optimal choice of $s$ is the smallest feasible support size.
\end{theorem}

\begin{proof}
By Proposition~\ref{prop:distortion}, both $R_1(\lambda,s)$ and $R_2(\lambda,s)$ are nondecreasing in $s$, so the smallest feasible $s$ minimizes distortion.
\end{proof}

\section{Sparse Gaussian mechanisms}

\begin{definition}[Sparse Gaussian family]\label{def:sparse-gaussian-family}
Let $d:\X\times\Y\to [0,\infty)$ be a distortion function. For each $x\in\X$, let $S(x)\subseteq \Y$ be a nonempty support set. For $\sigma>0$, define
\[
G_\sigma(y\mid x)=\frac{e^{-d(x,y)^2/(2\sigma^2)}}{W_x}\1\{y\in S(x)\},
\]
where $W_x=\sum_{y'\in S(x)} e^{-d(x,y')^2/(2\sigma^2)}.$
We call this a \emph{sparse Gaussian mechanism}. Its sparsity is governed by the cardinality $|S(x)|$.
\end{definition}

\subsection{Exact privacy characterization}

\begin{theorem}
\label{thm:gaussian-pure}
The sparse Gaussian mechanism in Definition~\ref{def:sparse-gaussian-family} is pure $\varepsilon$-LDP if and only if the following conditions hold:
\begin{enumerate}[label=(\roman*)]
    \item for every $x,x'\in\X$, one has $S(x)=S(x')$;
    \item if all supports are equal to a common set $S\subseteq\Y$, then
    \[
    \sup_{x,x'\in\X}\sup_{y\in S}
    \left[
    \frac{d(x',y)^2-d(x,y)^2}{2\sigma^2}+\log\frac{W_{x'}}{W_x}
    \right]\le \varepsilon.
    \]
\end{enumerate}
\end{theorem}

\begin{proof}
The proof is similar to that of Theorem~\ref{thm:pure}. Under common support,
\[
\frac{G_\sigma(y\mid x)}{G_\sigma(y\mid x')}=
\exp\!\left(\frac{d(x',y)^2-d(x,y)^2}{2\sigma^2}\right)\cdot \frac{W_{x'}}{W_x},
\]
which yields the stated privacy-loss formula.
\end{proof}

\begin{definition}[Ordered Gaussian privacy defect]
For the sparse Gaussian family, define
\[
\Delta^{\mathrm G}_{\varepsilon}(x,x'):=\sum_{y\in\Y}\pos{G_\sigma(y\mid x)-e^{\varepsilon}G_\sigma(y\mid x')}.
\]
\end{definition}

\begin{theorem}
\label{thm:gaussian-approx-general}
The sparse Gaussian mechanism in Definition~\ref{def:sparse-gaussian-family} is $(\varepsilon,\delta)$-LDP if and only if
\[
\sup_{x,x'\in\X}\Delta^{\mathrm G}_{\varepsilon}(x,x')\le \delta.
\]
Moreover,
\begin{align*}
\Delta^{\mathrm G}_{\varepsilon}(x,x')
&=
\sum_{y\in S(x)\setminus S(x')}
\frac{e^{-d(x,y)^2/(2\sigma^2)}}{W_x}
\\
&+\!\!\!\!\!\!\!
\sum_{y\in S(x)\cap S(x')}
\pos{\frac{e^{-d(x,y)^2/(2\sigma^2)}}{W_x}-e^{\varepsilon}\frac{e^{-d(x',y)^2/(2\sigma^2)}}{W_{x'}}}.
\end{align*}
\end{theorem}

\begin{proof}
The proof is the same as Theorem~\ref{thm:approx-general}, replacing the Laplace kernel by the Gaussian kernel.
\end{proof}

\subsection{Radius truncation and sparsity tradeoffs}

\begin{definition}[Radius-truncated sparse Gaussian mechanism]
Let $r\in\mathbb N$ and define
\begin{align*}
G_{\sigma,r}(y\mid x)&=\frac{e^{-(x-y)^2/(2\sigma^2)}}{\Gamma_r(\sigma)}\1\{|x-y|\le r\},
\\
\Gamma_r(\sigma)&:=\sum_{k=-r}^{r} e^{-k^2/(2\sigma^2)}.
\end{align*}
Its support size is $s=2r+1$.
\end{definition}

\begin{definition}[Separation-dependent Gaussian privacy defect]
For $h\ge 0$, define
\[
\delta_h^{\mathrm G}(\varepsilon,\sigma,r):=
\sum_{y\in\Y}\pos{G_{\sigma,r}(y\mid 0)-e^{\varepsilon}G_{\sigma,r}(y\mid h)}.
\]
\end{definition}

\begin{theorem}
\label{thm:gaussian-radius-exact}
The radius-truncated sparse Gaussian mechanism is $(\varepsilon,\delta)$-LDP if and only if
\[
\max_{0\le h\le D}\delta_h^{\mathrm G}(\varepsilon,\sigma,r)\le \delta.
\]
Moreover, if $0\le h\le 2r$, then
\begin{align*}
\delta_h^{\mathrm G}(\varepsilon,\sigma,r)
&=
\frac{1}{\Gamma_r(\sigma)}\sum_{k=-r}^{h-r-1} e^{-k^2/(2\sigma^2)}
\\
&+
\frac{1}{\Gamma_r(\sigma)}\sum_{k=h-r}^{r}
\pos{e^{-k^2/(2\sigma^2)}-e^{\varepsilon}e^{-(k-h)^2/(2\sigma^2)}},
\end{align*}
and if $h>2r$, then
\[
\delta_h^{\mathrm G}(\varepsilon,\sigma,r)=1.
\]
\end{theorem}

\begin{proof}
The proof is identical to Theorem~\ref{thm:radius-exact} after replacing the Laplace kernel with the Gaussian kernel.
\end{proof}

\begin{corollary}
If $|\X|\ge 2$, then the radius-truncated sparse Gaussian mechanism is not pure $\varepsilon$-LDP for any finite $\varepsilon$.
\end{corollary}

\begin{proof}
Distinct inputs have distinct supports, so pure privacy fails by Theorem~\ref{thm:gaussian-pure}.
\end{proof}

\begin{definition}[$s$-sparse Gaussian mechanism]
For odd integers $s\ge 1$, let
\[
t:=\frac{s-1}{2}.
\]
Define
\begin{align*}
G_{\sigma,s}(y\mid x)&=\frac{e^{-(x-y)^2/(2\sigma^2)}}{\Gamma_t(\sigma)}\1\{|x-y|\le t\},
\\
\Gamma_t(\sigma)&=\sum_{k=-t}^{t} e^{-k^2/(2\sigma^2)}.
\end{align*}
\end{definition}

\begin{theorem}
Fix a privacy range $H\ge 0$. Nontrivial $(\varepsilon,\delta)$-LDP on range $H$ with $\delta<1$ is impossible for the $s$-sparse radius-truncated Gaussian mechanism unless
\[
s\ge H+1.
\]
\end{theorem}

\begin{proof}
The support geometry is identical to that of the radius-truncated Laplace family, so the proof is the same as Theorem~\ref{thm:feasibility}.
\end{proof}

\begin{proposition}
\label{prop:gaussian-threshold}
Fix $0<h\le 2t$. Define
\[
\kappa_h:=\frac{h}{2}-\frac{\sigma^2\varepsilon}{h}.
\]
Then the overlap term in the Gaussian defect is positive only for indices $k<\kappa_h$.
\end{proposition}

\begin{proof}
For overlap index $k$, the overlap contribution is positive exactly when
\[
e^{-k^2/(2\sigma^2)}-e^{\varepsilon}e^{-(k-h)^2/(2\sigma^2)}>0,
\]
which is equivalent to
\[
\frac{(k-h)^2-k^2}{2\sigma^2}>\varepsilon.
\]
Since $(k-h)^2-k^2=h^2-2kh$, this becomes
\[
k<\frac{h}{2}-\frac{\sigma^2\varepsilon}{h}.
\]
\end{proof}

\begin{theorem}
\label{thm:gaussian-clean-bound}
Fix a privacy range $H\ge 0$ and let $t=(s-1)/2$. Suppose that
\[
\varepsilon\ge \frac{H(2t-H)}{2\sigma^2}
\qquad\text{and}\qquad
s\ge 2H+1.
\]
Then for every $0\le h\le H$,
\[
\delta_h^{\mathrm G}(\varepsilon,\sigma,s)=\frac{1}{\Gamma_t(\sigma)}\sum_{j=t-h+1}^{t} e^{-j^2/(2\sigma^2)},
\]
and hence
\begin{align*}
\delta^{\mathrm G,*}(\varepsilon,\sigma,s;H):=\max_{0\le h\le H}\delta_h^{\mathrm G}(\varepsilon,\sigma,s)
&\le \frac{H e^{-(t-H+1)^2/(2\sigma^2)}}{\Gamma_t(\sigma)}
\\&\le H e^{-(t-H+1)^2/(2\sigma^2)}.
\end{align*}
\end{theorem}

\begin{proof}
For fixed $h\le H$ and overlap index $k$, the pointwise Gaussian privacy loss is
\[
L_h^{\mathrm G}(k)=\log\frac{G_{\sigma,s}(k\mid 0)}{G_{\sigma,s}(k\mid h)}=
\frac{(k-h)^2-k^2}{2\sigma^2}=\frac{h^2-2kh}{2\sigma^2}.
\]
Because $k\ge h-t$, we have
\[
L_h^{\mathrm G}(k)\le \frac{h(2t-h)}{2\sigma^2}\le \frac{H(2t-H)}{2\sigma^2}\le \varepsilon.
\]
Therefore the overlap positive-part term vanishes for all $h\le H$. Only the support-leakage term remains, which yields the displayed formula and bound.
\end{proof}

\begin{corollary}
\label{cor:gaussian-window}
A sufficient condition for the $s$-sparse Gaussian mechanism to satisfy $(\varepsilon,\delta)$-LDP on range $H$ is that
\[
H-1+\sqrt{2\sigma^2\log\frac{H}{\delta}}\le t\le \frac{H}{2}+\frac{\sigma^2\varepsilon}{H},
\qquad t=\frac{s-1}{2}.
\]
Equivalently,
\[
2H-1+2\sqrt{2\sigma^2\log\frac{H}{\delta}}\le s\le H+1+\frac{2\sigma^2\varepsilon}{H}.
\]
\end{corollary}

\begin{proof}
The lower bound on $t$ is equivalent to
\[
H e^{-(t-H+1)^2/(2\sigma^2)}\le \delta,
\]
which is sufficient by Theorem~\ref{thm:gaussian-clean-bound}. The upper bound is equivalent to the overlap condition in the same theorem.
\end{proof}

\begin{proposition}
\label{prop:gaussian-distortion}
Let $t=(s-1)/2$. Then
\begin{align*}
R_1^{\mathrm G}(\sigma,s)&=\frac{2\sum_{j=1}^{t} j e^{-j^2/(2\sigma^2)}}{\Gamma_t(\sigma)},
\\
R_2^{\mathrm G}(\sigma,s)&=\frac{2\sum_{j=1}^{t} j^2 e^{-j^2/(2\sigma^2)}}{\Gamma_t(\sigma)}.
\end{align*}
Moreover, for fixed $\sigma$, both moments are nondecreasing in $s$.
\end{proposition}

\begin{proof}
The proof is the same as in the discrete-Laplace case.
\end{proof}

\begin{theorem}
Fix $\sigma>0$, a privacy range $H$, and target privacy parameters $(\varepsilon,\delta)$. Among all support sizes $s$ such that the $s$-sparse Gaussian mechanism satisfies $(\varepsilon,\delta)$-LDP on range $H$, the distortion-optimal choice of $s$ is the smallest feasible support size.
\end{theorem}

\begin{proof}
The Gaussian distortion moments are nondecreasing in $s$, so the smallest feasible support size minimizes distortion.
\end{proof}

\begin{remark}
\label{rem:laplace-vs-gaussian}
For the Laplace family, the overlap privacy loss can be eliminated by the simple condition $\lambda H\le \varepsilon$, after which enlarging the support mainly suppresses support leakage. For the Gaussian family, the overlap loss depends quadratically on the support radius and therefore cannot be ignored as easily: increasing $s$ reduces support leakage but also enlarges the region on which the Gaussian centers separate. This creates a sharper privacy--sparsity tradeoff.
\end{remark}

\section{Numerical illustrations}\label{sec:numerics}

Here, we present several numerical illustrations of the tradeoffs obtained above. Throughout this section, the reported values are obtained by direct evaluation of the exact privacy-defect formulas from Theorems~\ref{thm:radius-exact} and~\ref{thm:gaussian-radius-exact}, together with the exact distortion formulas from Proposition~\ref{prop:distortion} and Proposition~\ref{prop:gaussian-distortion}. Thus the tables and figures below are not Monte Carlo simulations; they are exact evaluations of the derived expressions.

For the Laplace family, we report
\[
\delta^*(\varepsilon,\lambda,s;H):=\max_{1\le h\le H}\delta_h(\varepsilon,\lambda,s),
\]
and for the Gaussian family we report
\[
\delta^{\mathrm G,*}(\varepsilon,\sigma,s;H):=\max_{1\le h\le H}\delta_h^{\mathrm G}(\varepsilon,\sigma,s).
\]

\subsection{Support-size sweeps}

Table~\ref{tab:laplace-sweep-s} fixes $(\varepsilon,\lambda,H)=(1,0.5,3)$ and varies the support size $s$. The first row illustrates the feasibility threshold from Theorem~\ref{thm:feasibility}: when $s=3$, one has $s-1=2<H=3$, and the exact privacy defect is $1$. Once the support becomes large enough to create overlap at the full privacy range, the privacy defect decreases while both distortion moments increase.

\begin{table}[ht]
\centering
\small
\begin{tabular}{cccc}
\toprule
$s$ & $\delta^*(1,0.5,s;3)$ & $R_1(0.5,s)$ & $R_2(0.5,s)$ \\
\midrule
3  & 1.0000 & 0.5481 & 0.5481 \\
5  & 0.6696 & 0.9104 & 1.4094 \\
7  & 0.4686 & 1.1851 & 2.4071 \\
9  & 0.3706 & 1.3929 & 3.4108 \\
11 & 0.3179 & 1.5475 & 4.3362 \\
13 & 0.2880 & 1.6603 & 5.1386 \\
\bottomrule
\end{tabular}
\caption{Laplace support-size sweep at $(\varepsilon,\lambda,H)=(1,0.5,3)$. Privacy improves with $s$, while distortion worsens with $s$.}
\label{tab:laplace-sweep-s}
\end{table}

\begin{figure}[ht]
\centering
\includegraphics[width=0.46\textwidth]{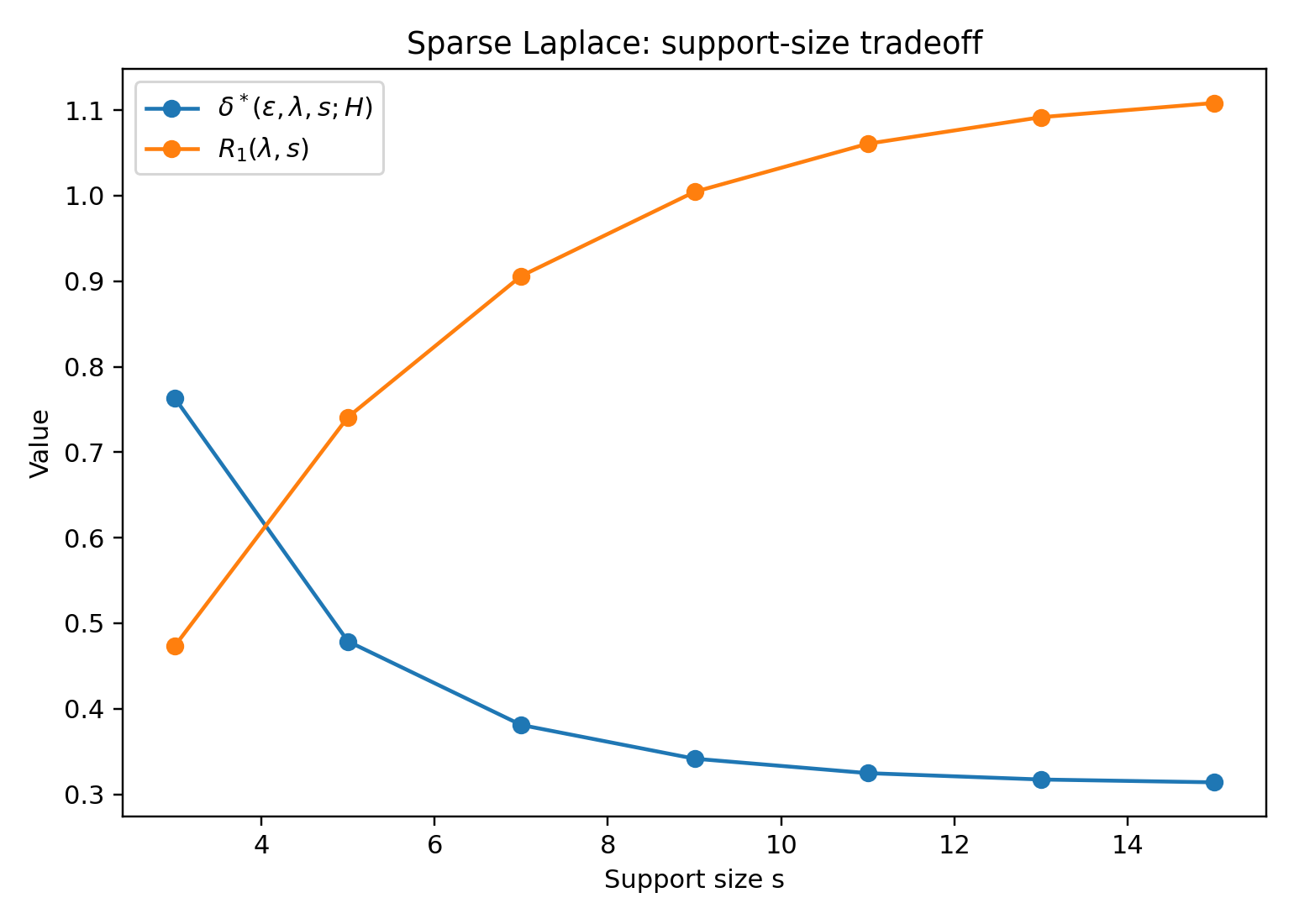}
\caption{Laplace support-size tradeoff for $(\varepsilon,\lambda,H)=(1,0.8,2)$. The exact privacy defect decreases with $s$, while the first distortion moment increases.}
\label{fig:laplace-support}
\end{figure}

The table and figure make the basic privacy--utility tension visible: larger supports reduce support leakage but increase distortion. This numerically confirms the design logic behind Theorem~\ref{thm:optimal-sparsity}.

For the Gaussian family, Table~\ref{tab:gaussian-sweep-s} fixes $(\varepsilon,\sigma,H)=(1,2,3)$ and varies the support size $s$. The same feasibility threshold appears: for $s=3$, privacy fails completely at range $H=3$, whereas larger supports reduce the exact privacy defect. As in the Laplace case, distortion increases with support size.

\begin{table}[ht]
\centering
\small
\begin{tabular}{cccc}
\toprule
$s$ & $\delta^{\mathrm G,*}(1,2,s;3)$ & $R_1^{\mathrm G}(2,s)$ & $R_2^{\mathrm G}(2,s)$ \\
\midrule
3  & 1.0000 & 0.6383 & 0.6383 \\
5  & 0.6257 & 1.0536 & 1.6634 \\
7  & 0.4173 & 1.3267 & 2.6929 \\
9  & 0.3468 & 1.4744 & 3.4283 \\
11 & 0.3255 & 1.5365 & 3.8084 \\
13 & 0.3203 & 1.5563 & 3.9513 \\
15 & 0.3193 & 1.5611 & 3.9906 \\
\bottomrule
\end{tabular}
\caption{Gaussian support-size sweep at $(\varepsilon,\sigma,H)=(1,2,3)$. Privacy improves with $s$, while distortion worsens with $s$.}
\label{tab:gaussian-sweep-s}
\end{table}

\begin{figure}[ht]
\centering
\includegraphics[width=0.46\textwidth]{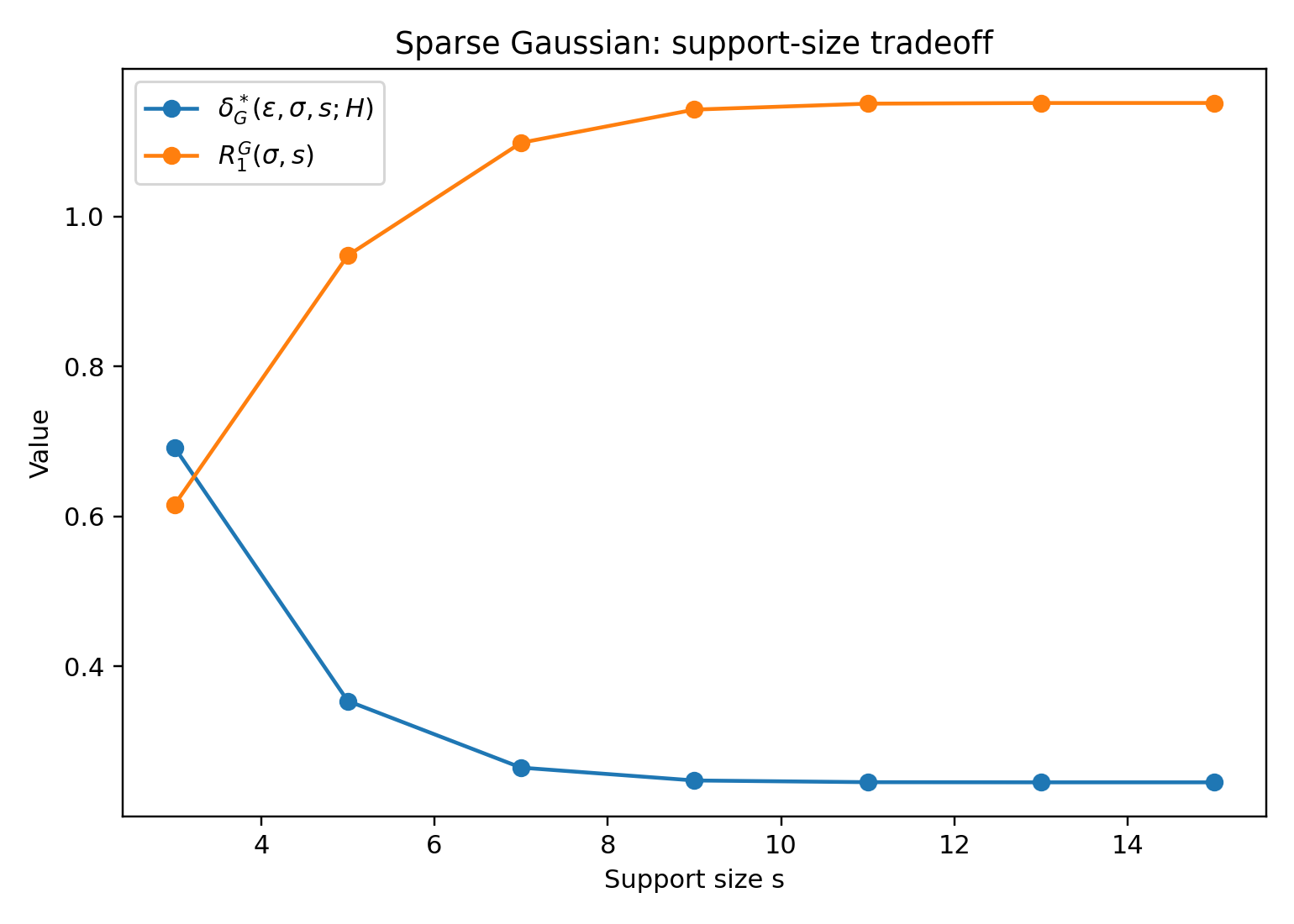}
\caption{Gaussian support-size tradeoff for $(\varepsilon,\sigma,H)=(1,1.5,2)$. The same privacy--distortion tension appears, but privacy decreases more slowly with $s$ than in the Laplace case.}
\label{fig:gaussian-support}
\end{figure}

In comparison with Table~\ref{tab:laplace-sweep-s}, the Gaussian defect decreases more slowly in this regime, which is consistent with the sharper overlap effects described in Remark~\ref{rem:laplace-vs-gaussian}.
\subsection{Parameter sweeps at fixed sparsity}

Table~\ref{tab:laplace-sweep-lambda} fixes $(\varepsilon,H,s)=(1,2,7)$ and varies the concentration parameter $\lambda$ of the Laplace distribution. Distortion decreases monotonically as $\lambda$ increases, because larger $\lambda$ concentrates more mass near the true input. By contrast, the privacy defect is non-monotone: it is minimized near $\lambda\approx 0.4$ in this example. For small $\lambda$, the mechanism is too diffuse and support leakage remains substantial; for large $\lambda$, the overlap privacy loss becomes dominant.

\begin{table}[ht]
\centering
\small
\begin{tabular}{cccc}
\toprule
$\lambda$ & $\delta^*(1,\lambda,7;2)$ & $R_1(\lambda,7)$ & $R_2(\lambda,7)$ \\
\midrule
0.2 & 0.2402 & 1.4996 & 3.3254 \\
0.4 & 0.1954 & 1.2872 & 2.6959 \\
0.6 & 0.2466 & 1.0870 & 2.1390 \\
0.8 & 0.3811 & 0.9061 & 1.6695 \\
1.0 & 0.4985 & 0.7483 & 1.2890 \\
1.2 & 0.5974 & 0.6142 & 0.9899 \\
\bottomrule
\end{tabular}
\caption{Laplace concentration sweep at $(\varepsilon,H,s)=(1,2,7)$. The exact privacy defect is non-monotone in $\lambda$, while distortion decreases with $\lambda$.}
\label{tab:laplace-sweep-lambda}
\end{table}

\begin{figure}[ht]
\centering
\includegraphics[width=0.46\textwidth]{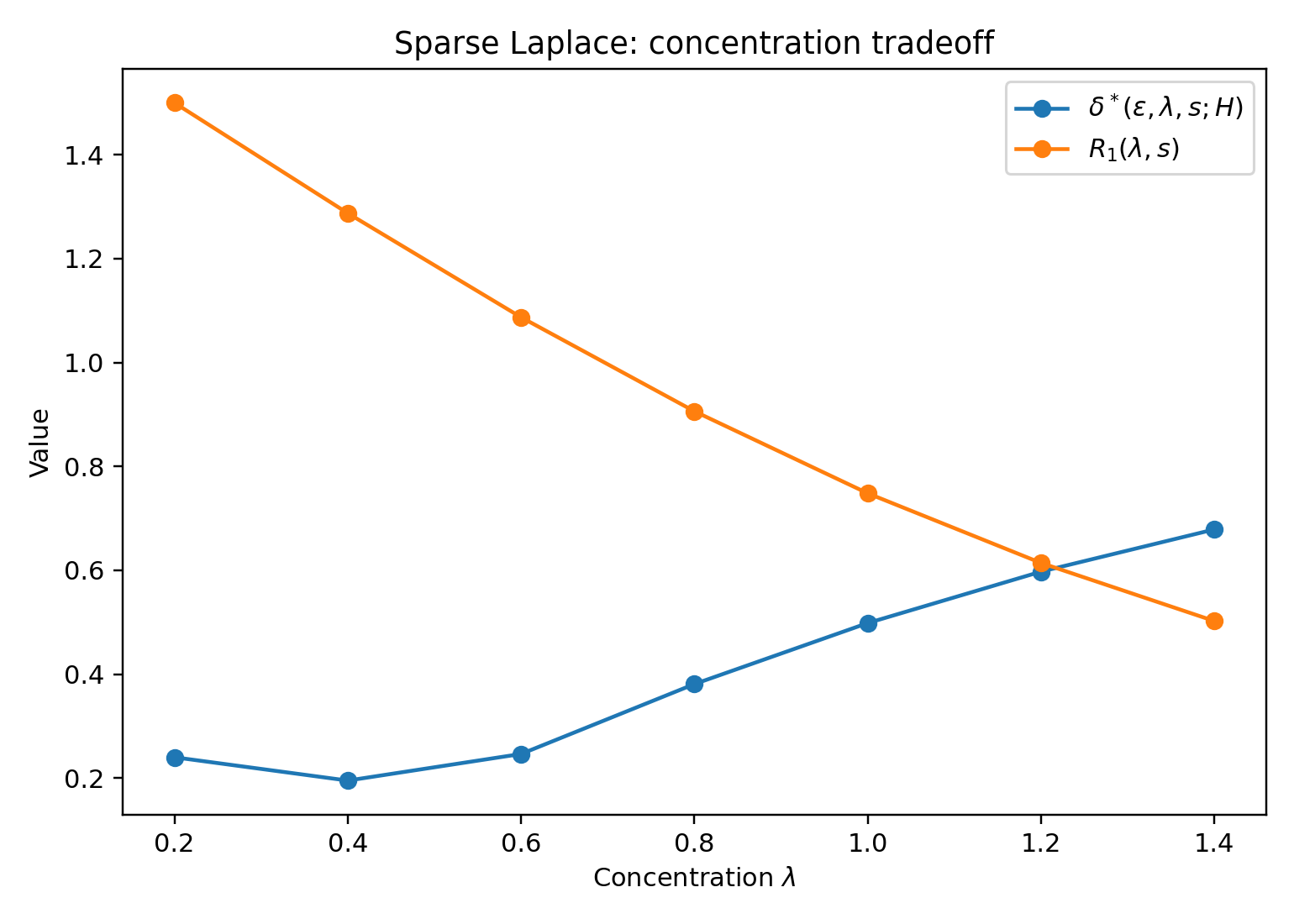}
\caption{Laplace concentration sweep. At fixed sparsity, increasing $\lambda$ improves utility but can worsen privacy after a point.}
\label{fig:laplace-lambda}
\end{figure}

This example shows that the support size is not the only design variable. Even with fixed sparsity, one must still balance concentration against privacy. In particular, minimizing distortion by taking a very large $\lambda$ may substantially worsen privacy.

Table~\ref{tab:gaussian-sweep-sigma} fixes $(\varepsilon,H,s)=(1,2,7)$ and varies the Gaussian scale $\sigma$. The privacy defect is again non-monotone: in this example it is smallest near $\sigma\approx 2$. Small $\sigma$ makes the mechanism highly concentrated and creates large privacy loss on the overlap; large $\sigma$ diffuses the mechanism and increases distortion without further improving privacy.

\begin{table}[ht]
\centering
\small
\begin{tabular}{cccc}
\toprule
$\sigma$ & $\delta^{\mathrm G,*}(1,\sigma,7;2)$ & $R_1^{\mathrm G}(\sigma,7)$ & $R_2^{\mathrm G}(\sigma,7)$ \\
\midrule
0.8 & 0.6886 & 0.5469 & 0.6398 \\
1.0 & 0.5407 & 0.7267 & 0.9959 \\
1.2 & 0.4009 & 0.8915 & 1.3997 \\
1.5 & 0.2651 & 1.0984 & 1.9831 \\
2.0 & 0.2012 & 1.3267 & 2.6929 \\
2.5 & 0.2301 & 1.4551 & 3.1140 \\
3.0 & 0.2466 & 1.5306 & 3.3673 \\
\bottomrule
\end{tabular}
\caption{Gaussian scale sweep at $(\varepsilon,H,s)=(1,2,7)$. The exact privacy defect is non-monotone in $\sigma$, while distortion increases with $\sigma$.}
\label{tab:gaussian-sweep-sigma}
\end{table}

\begin{figure}[ht]
\centering
\includegraphics[width=0.46\textwidth]{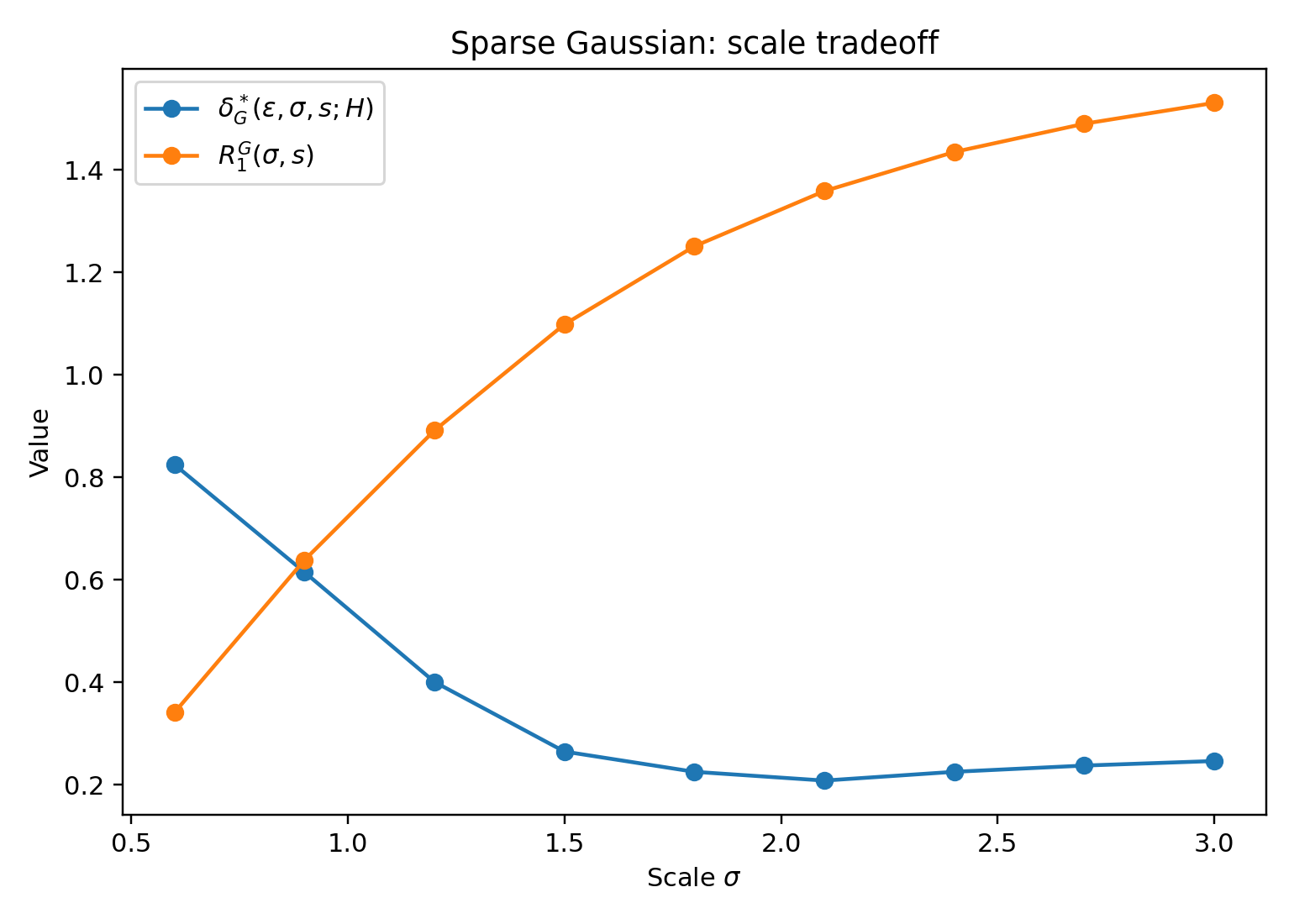}
\caption{Gaussian scale sweep. The privacy defect exhibits an interior minimizer, highlighting the two-parameter design nature of the Gaussian family.}
\label{fig:gaussian-sigma}
\end{figure}

This table illustrates a key difference from the Laplace case: the Gaussian family naturally exhibits an interior scale that balances support leakage and overlap loss. Thus the Gaussian design problem is genuinely two-dimensional, involving both support size and scale.


\section{Conclusion}

In this paper, we studied sparse locally private mechanisms in
which the conditional output law is supported on a small
admissible set $S(x)$ that may depend on the private input
$x$. Our main focus was on two concrete families: sparse
discrete-Laplace mechanisms and sparse Gaussian mechanisms.
For both families, we derived exact characterizations of pure
$\varepsilon$-LDP and approximate $(\varepsilon,\delta)$-LDP,
thereby making explicit how privacy depends on support
mismatch and overlap between neighboring conditional output
distributions.

A central structural conclusion of the paper is that
input-dependent sparse supports with pure
local privacy are obtained when all supports coincide. By contrast, in the
approximate privacy regime, sparsity becomes meaningful and
useful: the privacy defect decomposes naturally into a
support-leakage term and an overlap excess-loss term, so that
the $\delta$-budget admits a clear interpretation as the price
paid for input-dependent support truncation. This provides a
precise mechanism-level explanation for why sparsity is
compatible with $(\varepsilon,\delta)$-LDP but not with pure
$\varepsilon$-LDP.

For the radius-truncated models, we showed that the support
size $s=|S(x)|$ acts as an intrinsic complexity parameter of
the mechanism. On the one hand, nontrivial approximate
privacy requires a minimum degree of support overlap, which
imposes a lower bound on the feasible support size. On the
other hand, increasing the support reduces support leakage but
increases distortion. This leads to a simple and useful design
principle: among all support sizes satisfying the target privacy
constraint, the distortion-optimal choice is the smallest
feasible one. In the sparse discrete-Laplace case, this tradeoff
admits particularly clean sufficient bounds, while in the sparse
Gaussian case the overlap term introduces an additional
quadratic dependence on the support radius, resulting in a
sharper privacy--sparsity tension.

The numerical examples further illustrate that these tradeoffs
are not merely qualitative. The exact privacy defect can vary
non-monotonically with the mechanism parameters, and the
best design may depend jointly on the support size and the
concentration or scale parameter. In particular, the examples
show that support size alone does not determine performance:
even at fixed sparsity, privacy and distortion must still be
balanced through the underlying kernel parameter.


\clearpage 
\bibliographystyle{IEEEtran}
\bibliography{IEEEabrv,ITW}

\begin{thebibliography}{10}
\providecommand{\url}[1]{#1}
\csname url@samestyle\endcsname
\providecommand{\newblock}{\relax}
\providecommand{\bibinfo}[2]{#2}
\providecommand{\BIBentrySTDinterwordspacing}{\spaceskip=0pt\relax}
\providecommand{\BIBentryALTinterwordstretchfactor}{4}
\providecommand{\BIBentryALTinterwordspacing}{\spaceskip=\fontdimen2\font plus
\BIBentryALTinterwordstretchfactor\fontdimen3\font minus
  \fontdimen4\font\relax}
\providecommand{\BIBforeignlanguage}[2]{{%
\expandafter\ifx\csname l@#1\endcsname\relax
\typeout{** WARNING: IEEEtran.bst: No hyphenation pattern has been}%
\typeout{** loaded for the language `#1'. Using the pattern for}%
\typeout{** the default language instead.}%
\else
\language=\csname l@#1\endcsname
\fi
#2}}
\providecommand{\BIBdecl}{\relax}
\BIBdecl

\bibitem{dwork2006calibrating}
C.~Dwork, F.~McSherry, K.~Nissim, and A.~Smith, ``Calibrating noise to
  sensitivity in private data analysis,'' in \emph{Theory of cryptography
  conference}.\hskip 1em plus 0.5em minus 0.4em\relax Springer, 2006, pp.
  265--284.

\bibitem{dwork2014book}
C.~Dwork and A.~Roth, ``The algorithmic foundations of differential privacy,''
  \emph{Foundations and trends{\textregistered} in theoretical computer
  science}, vol.~9, no. 3-4, pp. 211--487, 2014.

\bibitem{rappor2014}
{\'U}.~Erlingsson, V.~Pihur, and A.~Korolova, ``Rappor: Randomized aggregatable
  privacy-preserving ordinal response,'' in \emph{Proceedings of the 2014 ACM
  SIGSAC conference on computer and communications security}, 2014, pp.
  1054--1067.

\bibitem{acharya2021sparse}
J.~Acharya, P.~Kairouz, Y.~Liu, and Z.~Sun, ``Estimating sparse discrete
  distributions under privacy and communication constraints,'' in
  \emph{Algorithmic Learning Theory}.\hskip 1em plus 0.5em minus 0.4em\relax
  PMLR, 2021, pp. 79--98.

\bibitem{andres2013geo}
M.~E. Andr{\'e}s, N.~E. Bordenabe, K.~Chatzikokolakis, and C.~Palamidessi,
  ``Geo-indistinguishability: Differential privacy for location-based
  systems,'' in \emph{Proceedings of the 2013 ACM SIGSAC conference on Computer
  \& communications security}, 2013, pp. 901--914.

\bibitem{kacem2018geometric}
L.~Kacem and C.~Palamidessi, ``Geometric noise for locally private counting
  queries,'' in \emph{Proceedings of the 13th Workshop on Programming Languages
  and Analysis for Security}, 2018, pp. 13--16.

\bibitem{balle2018gaussian}
B.~Balle and Y.-X. Wang, ``Improving the gaussian mechanism for differential
  privacy: Analytical calibration and optimal denoising,'' in
  \emph{International conference on machine learning}.\hskip 1em plus 0.5em
  minus 0.4em\relax PMLR, 2018, pp. 394--403.

\bibitem{canonne2020discretegaussian}
C.~L. Canonne, G.~Kamath, and T.~Steinke, ``The discrete gaussian for
  differential privacy,'' \emph{Advances in Neural Information Processing
  Systems}, vol.~33, pp. 15\,676--15\,688, 2020.

\bibitem{warner1965randomized}
S.~L. Warner, ``Randomized response: A survey technique for eliminating evasive
  answer bias,'' \emph{Journal of the American Statistical Association},
  vol.~60, no. 309, pp. 63--69, 1965.

\bibitem{duchi2014local}
J.~C. Duchi, M.~I. Jordan, and M.~J. Wainwright, ``Local privacy and minimax
  bounds: Sharp rates for probability estimation,'' in \emph{Advances in Neural
  Information Processing Systems}, 2014.

\bibitem{duchi2018minimax}
------, ``Minimax optimal procedures for locally private estimation,''
  \emph{Journal of the American Statistical Association}, vol. 113, no. 521,
  pp. 182--201, 2018.

\bibitem{mcsherry2007mechanism}
F.~McSherry and K.~Talwar, ``Mechanism design via differential privacy,'' in
  \emph{48th Annual IEEE Symposium on Foundations of Computer Science
  (FOCS'07)}.\hskip 1em plus 0.5em minus 0.4em\relax IEEE, 2007, pp. 94--103.

\bibitem{geng2012optimal}
Q.~Geng and P.~Viswanath, ``The optimal mechanism in differential privacy,'' in
  \emph{2014 IEEE international symposium on information theory}.\hskip 1em
  plus 0.5em minus 0.4em\relax IEEE, 2014, pp. 2371--2375.

\bibitem{kairouz2014extremal}
P.~Kairouz, S.~Oh, and P.~Viswanath, ``Extremal mechanisms for local
  differential privacy,'' \emph{Advances in neural information processing
  systems}, vol.~27, 2014.

\bibitem{kairouz2016discrete}
P.~Kairouz, K.~Bonawitz, and D.~Ramage, ``Discrete distribution estimation
  under local privacy,'' in \emph{International Conference on Machine
  Learning}.\hskip 1em plus 0.5em minus 0.4em\relax PMLR, 2016, pp. 2436--2444.

\bibitem{agarwal2021skellam}
N.~Agarwal, P.~Kairouz, and Z.~Liu, ``The skellam mechanism for differentially
  private federated learning,'' \emph{Advances in neural information processing
  systems}, vol.~34, pp. 5052--5064, 2021.

\bibitem{chen2022poisson}
W.-N. Chen, A.~{\"O}zg{\"u}r, and P.~Kairouz, ``The poisson binomial mechanism
  for secure and private federated learning,'' \emph{arXiv preprint
  arXiv:2207.09916}, 2022.

\bibitem{murakami2019utility}
T.~Murakami, H.~Hino, J.~Sakuma, Y.~Kawamoto, and C.~Troncoso,
  ``Utility-optimized local differential privacy mechanisms for distribution
  estimation,'' in \emph{USENIX Security Symposium}, 2019.

\bibitem{Sadeghi_Chien_2024}
\BIBentryALTinterwordspacing
P.~Sadeghi and C.-H. Chien, ``On the connection between the abs perturbation
  methodology and differential privacy,'' \emph{Journal of Privacy and
  Confidentiality}, vol.~14, no.~2, Jul. 2024. [Online]. Available:
  \url{https://journalprivacyconfidentiality.org/index.php/jpc/article/view/859}
\BIBentrySTDinterwordspacing

\end{thebibliography}


\begin{thebibliography}{99}

\bibitem{dwork2006calibrating}
Cynthia Dwork, Frank McSherry, Kobbi Nissim, and Adam Smith.
\newblock Calibrating noise to sensitivity in private data analysis.
\newblock In \emph{Proceedings of the Third Theory of Cryptography Conference (TCC)}, pages 265--284, 2006.

\bibitem{dwork2014book}
Cynthia Dwork and Aaron Roth.
\newblock \emph{The Algorithmic Foundations of Differential Privacy}.
\newblock Foundations and Trends in Theoretical Computer Science, 9(3--4):211--407, 2014.

\bibitem{rappor2014}
\'Ulfar Erlingsson, Vasyl Pihur, and Aleksandra Korolova.
\newblock {RAPPOR}: Randomized aggregatable privacy-preserving ordinal response.
\newblock In \emph{Proceedings of the 21st ACM SIGSAC Conference on Computer and Communications Security}, pages 1054--1067, 2014.

\bibitem{acharya2021sparse}
Jayadev Acharya, Peter Kairouz, Yuhan Liu, and Ziteng Sun.
\newblock Estimating sparse discrete distributions under privacy and communication constraints.
\newblock In \emph{Proceedings of the 32nd International Conference on Algorithmic Learning Theory}, volume 132 of \emph{Proceedings of Machine Learning Research}, pages 79--98, 2021.

\bibitem{andres2013geo}
Miguel E. Andr\'es, Nicol\'as E. Bordenabe, Konstantinos Chatzikokolakis, and Catuscia Palamidessi.
\newblock Geo-indistinguishability: Differential privacy for location-based systems.
\newblock In \emph{Proceedings of the 2013 ACM SIGSAC Conference on Computer and Communications Security}, pages 901--914, 2013.

\bibitem{kacem2018geometric}
Lefki Kacem and Catuscia Palamidessi.
\newblock Geometric noise for locally private counting queries.
\newblock In \emph{Proceedings of the 13th Workshop on Programming Languages and Analysis for Security (PLAS@CCS)}, pages 13--16, 2018.

\bibitem{balle2018gaussian}
Borja Balle and Yu-Xiang Wang.
\newblock Improving the Gaussian mechanism for differential privacy: Analytical calibration and optimal denoising.
\newblock In \emph{Proceedings of the 35th International Conference on Machine Learning}, volume 80 of \emph{Proceedings of Machine Learning Research}, pages 394--403, 2018.

\bibitem{canonne2020discretegaussian}
Cl\'ement L. Canonne, Gautam Kamath, and Thomas Steinke.
\newblock The discrete Gaussian for differential privacy.
\newblock In \emph{Advances in Neural Information Processing Systems 33 (NeurIPS 2020)}, pages 15676--15688, 2020.


\bibitem{mcsherry2007mechanism}
Frank McSherry and Kunal Talwar.
\newblock Mechanism design via differential privacy.
\newblock In \emph{Proceedings of the 48th Annual IEEE Symposium on Foundations of Computer Science (FOCS)}, pages 94--103, 2007.

\bibitem{geng2012optimal}
Quan Geng and Pramod Viswanath.
\newblock The optimal mechanism in differential privacy.
\newblock \emph{arXiv preprint arXiv:1212.1186}, 2012.

\bibitem{kairouz2014extremal}
Peter Kairouz, Sewoong Oh, and Pramod Viswanath.
\newblock Extremal mechanisms for local differential privacy.
\newblock In \emph{Advances in Neural Information Processing Systems 27 (NeurIPS 2014)}, pages 2879--2887, 2014.

\bibitem{kairouz2016discrete}
Peter Kairouz, Sewoong Oh, and Pramod Viswanath.
\newblock Discrete distribution estimation under local privacy.
\newblock In \emph{Proceedings of the 33rd International Conference on Machine Learning}, volume 48 of \emph{Proceedings of Machine Learning Research}, pages 2436--2444, 2016.

\bibitem{agarwal2021skellam}
Naman Agarwal, Peter Kairouz, and Ziyu Liu.
\newblock The Skellam mechanism for differentially private federated learning.
\newblock In \emph{Advances in Neural Information Processing Systems 34 (NeurIPS 2021)}, pages 5052--5064, 2021.

\bibitem{chen2022poisson}
Wei-Ning Chen, Ayfer {\"O}zg{\"u}r, and Peter Kairouz.
\newblock The Poisson binomial mechanism for secure and private federated learning.
\newblock \emph{arXiv preprint arXiv:2207.09916}, 2022.

\end{thebibliography}
\end{document}